\RequirePackage{fix-cm}
\documentclass[twocolumn]{svjour3}          % twocolumn
\smartqed  % flush right qed marks, e.g. at end of proof
\usepackage{graphicx}
\usepackage{latexsym}
\usepackage{amsmath}
\usepackage{amssymb}
\journalname{Brazilian Journal of Physics}
\begin{document}

\title{Distribution of Mutual Information in Multipartite States}

\author{Jonas Maziero}

\institute{Jonas Maziero \at Universidade Federal do Pampa, Campus Bag\'e, 96413-170, Bag\'e, RS, Brazil \\ 
\email{jonas.maziero@ufsm.br}
\\ \emph{Present address:} Departamento de F\'isica, Universidade Federal de Santa Maria, 97105-900, Santa Maria, RS, Brazil}

\date{Received: date / Accepted: date}

\maketitle

\begin{abstract}
Using the relative entropy of total correlation, we derive an
expression relating the mutual information of $n$-partite pure states
to the sum of the mutual informations and entropies of its marginals
and analyze some of its implications. Besides, by utilizing the extended
strong subadditivity of von Neumann entropy, we obtain generalized
monogamy relations for the total correlation in three-partite mixed
states. These inequalities lead to a tight lower bound for this correlation
in terms of the sum of the bipartite mutual informations. We use this
bound to propose a measure for residual three-partite total correlation
and discuss the non-applicability of this kind of quantifier to measure
genuine multiparty correlations.


\end{abstract}

%\textbf{Keywords:} Distribution of correlations; Mutual information; Multipartite states; Monogamy relations.

\section{Introduction}
\label{sec1}

The correlations among the parts constituents of a system have been
at the central stage of discussions regarding fundamental concepts
of quantum physics since nearly a decade after its formulation \cite{EPR,Schcat}.
In quantum information science, the quantum part of correlations is
believed to be one of the main factors responsible by the so called
quantum advantage \cite{Jozsa-Linden,Datta-DQC1,Walther,Retamal,Madhok-ijmpb,Escher-bjp}.
However, classical correlation has also proven worthy of investigation.
For example, the derivative of bipartite classical correlation can
be used to indicate critical points of quantum phase transition \cite{Maziero-xy-pra,Sarandy-qpt-cc}.
Besides, the sudden change phenomenon of the classical correlation
between two qubits during its decoherent dynamics \cite{Maziero-SC,Auccaise-SC}
was shown to characterize the emergence of the pointer bases in a
quantum measurement process \cite{Cornelio-PBS}. 

With respect to multipartite systems, researches concerning its correlations
are important both from the practical point of view (e.g. because
of large scale implementations of protocols in quantum information
science) and also for the foundations of physics (e.g. understanding
the rising of collective behavior is essential in investigations of
quantum and classical phase transitions). The structure of the correlations presented 
in general multiparty states has been investigated using different techniques in 
Refs. \cite{LindenPW,LindenW,MBerry,Cavalcanti,Zhou,Walk,Parashar}.
In this article, considering finite dimensional systems and using relative entropy-based 
measures of correlations, we address some instances of the problem
of distribution of total correlation---which encompasses both the
classical and quantum ones---in multi-particle systems.

In what concern the quantification of total correlation, for bipartite
states a well justified (both physically \cite{Groisman-QMI} and
operationally \cite{Schumacher-QMI}) measure for total correlation
is obtained via a direct generalization of Shannon's classical mutual
information \cite{Shannon-Weaver}. This quantifier is dubbed quantum
mutual information and is defined as: 
\begin{equation}
I(\rho_{12})=S(\rho_{1})+S(\rho_{2})-S(\rho_{12}),
\end{equation}
with $\rho_{s}$ being the density operator (i.e., $\rho_{s}\ge0$
and $\mathrm{tr}(\rho_{s})=1$) on the Hilbert's space $\mathcal{H}_{s}$
of system $s$ ($\rho_{s}\in\mathcal{D}(\mathcal{H}_{s})$) and $S(\rho_{s})=-\mathrm{tr}(\rho_{s}\log_{2}\rho_{s})$
being its von Neumann's entropy. Above and hereafter $\rho_{s}=\mathrm{tr}_{s'}(\rho_{ss'})$
is the reduced state of subsystem $s$, obtained by tracing out the
other parties $s'$ of the whole system. 

On the other hand, for multipartite states the situation is less understood.
Venn's diagram-based approaches may lead to negative measures of correlation
\cite{Vedral-RE-RMP,Cover-Thomas}. One quantifier free from this
problem was introduced in Ref. \cite{Herbut-JPA} as follows:
\begin{equation}
I(\rho_{1\cdots n})=\sum_{s=1}^{n}S(\rho_{s})-S(\rho_{1\cdots n}).\label{eq:mQMI}
\end{equation}

In Ref. \cite{Modi-RE}, this issue was addressed by quantifying
the total correlation in a multipartite state by how distinguishable it is from uncorrelated (product) states. Using the quantum
relative entropy (QRE) \cite{Vedral-RE-RMP,Schumacher-Review,Wehrl},
\begin{equation}
S(\rho||\sigma)=\mathrm{tr}(\rho\log_{2}\rho)-\mathrm{tr}(\rho\log_{2}\sigma),
\label{eq:QRE}
\end{equation}
to measure distinguishability between any pair of quantum states $\rho$ and
$\sigma$, Modi \emph{et al.} showed that the closest (less distinguishable) \cite{QRE} product state
from any density operator is given by the states of its marginals
in the product form. The multipartite mutual information
defined in this way is called \emph{relative entropy of total correlation}
(RETC) and is given as in Eq. (\ref{eq:mQMI}). 

Here we give a simple, alternative, proof for the result obtained by Modi \emph{et al.} in Ref. \cite{Modi-RE}. 
\begin{proposition}
The closest product state of any multipartite state $\rho_{1\cdots n}$ is obtained from its marginals in the product form.
\end{proposition}
\begin{proof}
Let $\bigotimes_{s=1}^{n}\sigma_{s}$ be any $n$-partite product state. Using the definition of QRE in Eq. (\ref{eq:QRE}), one can write 
(we postpone the proof of this equality to Appendix A):
\begin{equation}
S(\rho_{1\cdots n}||{\bigotimes_{s=1}^{n}}\sigma_{s})=S(\rho_{1\cdots n}||{\bigotimes_{s=1}^{n}}\rho_{s})
+{\sum_{s=1}^{n}}S(\rho_{s}||\sigma_{s}). 
\label{QRE-equality}
\end{equation}
As $S(\rho||\sigma)\ge0$ with equality if and only if $\rho=\sigma$,
we see that
\begin{equation}
S(\rho_{1\cdots n}||\bigotimes_{s=1}^{n}\sigma_{s})\ge S(\rho_{1\cdots n}||\bigotimes_{s=1}^{n}\rho_{s}),
\end{equation}
with equality obtained only if $\sum_{s=1}^{n}S(\rho_{s}||\sigma_{s})=0$,
i.e., if $\rho_{s}=\sigma_{s}\mbox{ }\forall s$. Therefore
\begin{eqnarray}
I(\rho_{1\cdots n})&=&\min_{\bigotimes_{s=1}^{n}\sigma_{s}} S(\rho_{1\cdots n}||{\textstyle \bigotimes_{s=1}^{n}}\sigma_{s}) \nonumber \\
&=&S(\rho_{1\cdots n}||{\textstyle \bigotimes_{s=1}^{n}}\rho_{s}),
\end{eqnarray}
concluding thus the proof of the proposition. \qed
\end{proof}

In the subsequent sections we shall regard the distribution of the
RETC in $n$-partite pure (Sec. \ref{sec2}) and three-partite mixed
(Sec. \ref{sec3}) states.

\section{Distribution of mutual information in $n$-partite pure states}
\label{sec2}

Let us consider the case of a system with $n$ parties in a pure state
$|\psi_{1\cdots n}\rangle\in\mathcal{H}_{1}\otimes\cdots\otimes\mathcal{H}_{n}$,
with $\mathcal{H}_{s}$ being the state space for the subsystem $s$
($s=1,\cdots,n$). Using the definition presented in the Sec. \ref{sec1}
and noting that the uncertainty associated with the preparation of
a pure state is null ($S(|\psi_{1\cdots n}\rangle)=0$), one see that
the total mutual information of $n$ subsystems in a pure state $|\psi_{1\cdots n}\rangle$
is given by:
\begin{equation}
I(|\psi_{1\cdots n}\rangle)={\textstyle \sum_{s=1}^{n}}S(\rho_{s}).\label{eq:InPS}
\end{equation}
Below this correlation is related to the mutual informations and entropies
of the marginals of $|\psi_{1\cdots n}\rangle$. First, we shall define
some quantities to be used in the sequence of the article.

\begin{definition}

\label{def1}

The sum of the mutual informations of the $(n-k)$-partite reductions
of $\rho_{1\cdots n}$ is defined as $\mathcal{I}_{n-k}(\rho_{1\cdots n})$.

\end{definition}

\begin{definition}

The sum of the entropies of the $k$-partite reductions of $\rho_{1\cdots n}$
is defined as $\mathcal{S}_{k}(\rho_{1\cdots n})$.

\end{definition}

We observe that the Definition \ref{def1} only makes sense if $n-k\ge2$,
i.e., if $k=1,2,\cdots,n-2$. In what follows we will use these definitions
in the proof for the following proposition.

\begin{proposition}

The total amount of information shared among $n$ parties in an pure
state $|\psi_{1\cdots n}\rangle$ can be written as
\begin{equation}
I(|\psi_{1\cdots n}\rangle)=\frac{k!\left(\mathcal{I}_{n-k}(|\psi_{1\cdots n}\rangle)+\mathcal{S}_{k}(|\psi_{1\cdots n}\rangle)\right)}{\prod_{i=1}^{k}(n-i)}.
\label{eq:gen-rel}
\end{equation}

\end{proposition}

\begin{proof}

For the purpose of proving this proposition, it will be helpful first
to split the system $\mathcal{H}_{1}\otimes\cdots\otimes\mathcal{H}_{n}$
in two components $\mathcal{H}_{x}\otimes\mathcal{H}_{\overline{x}}$,
with $x$ denoting one or more parties and $\overline{x}$ being the
rest of the system. It follows then, from Schmidt's decomposition
(see e.g. Ref. \cite{Nielsen-Chuang}), that $S(\rho_{x})=S(\rho_{\overline{x}})$,
with $\rho_{x(\overline{x})}=\mathrm{tr}_{\overline{x}(x)}(|\psi_{1\cdots n}\rangle\langle\psi_{1\cdots n}|)$. 

Now, let us begin by computing the sum of the mutual informations
of the $(n-1)$-partite reductions of $|\psi_{1\cdots n}\rangle$.
In this case $k=1$ and
\begin{eqnarray}
\mathcal{I}_{n-1}(|\psi_{1\cdots n}\rangle) & = & \sum_{s=1}^{n}I(\rho_{\overline{s}})\nonumber \\
 & = & \sum_{s=1}^{n}\sum_{\underset{s'\ne s}{s'=1}}^{n}S(\rho_{s'})-\sum_{s=1}^{n}S(\rho_{\overline{s}}).
\end{eqnarray}
There are $n(n-1)$ one-party entropies in the first term on the right
hand side of the last equality. As the $n$ subsystems appear with
the same frequency in this term, and $S(\rho_{\overline{s}})=S(\rho_{s})$,
we get
\begin{equation}
\mathcal{I}_{n-1}=\frac{(n-1)}{1!}\mathcal{S}_{1}(|\psi_{1\cdots n}\rangle)-\mathcal{S}_{1}(|\psi_{1\cdots n}\rangle).
\end{equation}

For $k=2$, if we take the sum of the total correlation of the $(n-2)$-partite
marginals of $|\psi_{1\cdots n}\rangle$, we obtain
\begin{eqnarray}
\mathcal{I}_{n-2} & = & \sum_{s=1}^{n}\sum_{s'=s+1}^{n}I(\rho_{\overline{ss'}})\nonumber \\
 & = & \sum_{s=1}^{n}\sum_{s'=s+1}^{n}\sum_{\underset{s'\ne s,s'}{s''=1}}^{n}S(\rho_{s''})-\sum_{s=1}^{n}\sum_{s'=s+1}^{n}S(\rho_{\overline{ss'}})\nonumber \\
 & = & \frac{(n-1)(n-2)}{2!}\mathcal{S}_{1}(|\psi_{1\cdots n}\rangle)-\mathcal{S}_{2}(|\psi_{1\cdots n}\rangle).
\end{eqnarray}
In order to obtain the last equality, and below, we note that there
are $n!/(k!(n-k)!)$ different reductions of $|\psi_{1\cdots n}\rangle$
comprising $k$ parties. 

For $k=3$, the sum of $(n-3)$-partite mutual informations is
\begin{eqnarray}
\mathcal{I}_{n-3} & = & \sum_{s=1}^{n}\sum_{s'=s+1}^{n}\sum_{s''=s'+1}^{n}I(\rho_{\overline{ss's''}})\nonumber \\
 & = & \sum_{s=1}^{n}\sum_{s'=s+1}^{n}\sum_{s''=s'+1}^{n}\sum_{\underset{s'''\ne s,s',s''}{s'''=1}}^{n}S(\rho_{s'''})-\mathcal{S}_{3}\\
 & = & \frac{(n-1)(n-2)(n-3)}{3!}\mathcal{S}_{1}(|\psi_{1\cdots n}\rangle)-\mathcal{S}_{3}(|\psi_{1\cdots n}\rangle).\nonumber 
\end{eqnarray}

So, by inductive reasoning, one see that for any value of $k$ in
the set $\{1,2,\cdots,n-2\}$, the following equality holds
\begin{equation}
\mathcal{I}_{n-k}=\frac{\prod_{i=1}^{k}(n-i)}{k!}\mathcal{S}_{1}(|\psi_{1\cdots n}\rangle)-\mathcal{S}_{k}(|\psi_{1\cdots n}\rangle).
\end{equation}
Once $I(|\psi_{1\cdots n}\rangle)=\mathcal{S}_{1}(|\psi_{1\cdots n}\rangle)$,
this equation is seen to be equivalent to Eq. (\ref{eq:gen-rel}),
concluding thus the proof of the proposition. \qed

\end{proof}

\subsection{Some Particular Cases}

\label{sec2a}

Now we regard the particular case in which $k=1$ (and therefore $n\ge3$).
It follows from Eq. (\ref{eq:gen-rel}) that
\begin{equation}
I(|\psi_{1\cdots n}\rangle)=\frac{\mathcal{I}_{n-1}(|\psi_{1\cdots n}\rangle)}{(n-2)\label{eq:InIn-1}}.
\end{equation}
Thus, for three-particle pure states ($n=3$), the following equality
is obtained:
\begin{equation}
I(|\psi_{123}\rangle)=\mathcal{I}_{2}(|\psi_{123}\rangle).\label{eq:eq-3ps}
\end{equation}
So the total correlation in $|\psi_{123}\rangle$ is shown to be equal
to the sum of the mutual informations of its bipartite marginals.
As an example let us consider the three-qubit Greenberger-Horne-Zeilinger
state \cite{GHZ}: $|\mathrm{GHZ}_{3}\rangle=2^{-1/2}(|0_{1}0_{2}0_{3}\rangle+|1_{1}1_{2}1_{3}\rangle)$,
where $\{|0_{s}\rangle,|1_{s}\rangle\}$ is the computational basis
for the qubit $s$. In the last equation and throughout the article
we shall use the notation: $|\psi_{s}\rangle\otimes|\phi_{s'}\rangle=|\psi_{s}\phi_{s'}\rangle=|\psi\phi\rangle$.
The reduced states of $|\mathrm{GHZ}_{3}\rangle$ are $\rho_{ss'}=2^{-1}(|0_{s}0_{s'}\rangle\langle0_{s}0_{s'}|+|1_{s}1_{s'}\rangle\langle1_{s}1_{s'}|)$,
with $ss'=12,13,23$, and $\rho_{r}=2^{-1}(|0_{r}\rangle\langle0_{r}|+|1_{r}\rangle\langle1_{r}|)$,
with $r=1,2,3$. Thus $S(\rho_{ss'})=S(\rho_{r})=1$, which leads
to $I(|\mathrm{GHZ}_{3}\rangle)=3$ and $I(\rho_{ss'})=1$, and consequently
to the equalities in Eqs. (\ref{eq:InIn-1}) and (\ref{eq:eq-3ps}).

On the other side, one see that for $n\ge4$ the sum of the $(n-1)$-partite
reductions' total correlation overestimate the total information shared
among the $n$ subsystems, i.e.,
\begin{equation}
I(|\psi_{1\cdots n}\rangle)<\mathcal{I}_{n-1}(|\psi_{1\cdots n}\rangle)\label{eq:ineq-4ps}
\end{equation}
for $n\ge 4$.
For the sake of exemplifying the applicability of this inequality
we regard again the example of Greenberger-Horne-Zeilinger states,
but for four qubits: $|\mathrm{GHZ}_{4}\rangle=2^{-1/2}(|0_{1}0_{2}0_{3}0_{4}\rangle+|1_{1}1_{2}1_{3}1_{4}\rangle)$.
The reduced states of $|\mathrm{GHZ}_{4}\rangle$ we need here are
the three-qubit density operators:
$
\rho_{ss's''} = 2^{-1}(|000\rangle\langle000| + |111\rangle\langle111|),
$
with $ss's''=123,124,133,234$, and the one-qubit density matrices:
$\rho_{r}=2^{-1}(|0_{r}\rangle\langle0_{r}|+|1_{r}\rangle\langle1_{r}|)$,
with $r=1,2,3,4$. Hence $S(\rho_{ss's''})=S(\rho_{r})=1$. Therefore it follows that
$I(|\mathrm{GHZ}_{4}\rangle)=4$ and $I(\rho_{ss's''})=2$. Using
these values we obtain the relations: 
\begin{equation}
\mathcal{I}_{3}(|\mathrm{GHZ}_{4}\rangle)=8=(4-2)I(|\mathrm{GHZ}_{4}\rangle)>I(|\mathrm{GHZ}_{4}\rangle),
\end{equation}
which satisfy the equality in Eq. (\ref{eq:InIn-1}) and the inequality
in Eq. (\ref{eq:ineq-4ps}).

\section{Distribution of mutual information in three-partite states}

\label{sec3}

\subsection{Generalized Monogamy Relations for Total Correlation}

One of the most important inequalities in quantum information theory
is the strong subadditivity property of von Neumann entropy (SSA),
by which \cite{Lieb-Ruskai}:
\begin{equation}
S(\rho_{ss'})+S(\rho_{ss''})-S(\rho_{s})-S(\rho_{123})\ge0,\label{eq:ss}
\end{equation}
where $ss's''=123,231,321$. 

Recently, Carlen and Lieb proved
an extended version for the SSA (ESSA) \cite{ESSA}:
\begin{eqnarray}
 &  & S(\rho_{ss'})+S(\rho_{ss''})-S(\rho_{s})-S(\rho_{123})\ge\nonumber \\
 &  & 2\max\{S(\rho_{s'})-S(\rho_{s's''}),S(\rho_{s''})-S(\rho_{s's''}),0\}
\end{eqnarray}
We shall use this inequality to prove the following proposition.

\begin{proposition}

The total mutual information of three-partite states imposes the following
constraint for the correlations of its bipartite marginals:
\begin{eqnarray}
I(\rho_{123}) & \ge & I(\rho_{ss'})+I(\rho_{ss''})+\label{eq:gmr1-i}\\
 &  & 2\max\{I(\rho_{s's''})-S(\rho_{s'}),I(\rho_{s's''})-S(\rho_{s''}),0\}\nonumber 
\end{eqnarray}
with $ss's''=123,231,321$.

\end{proposition}

\begin{proof}

Let us begin the proof by rewriting the SSA in Eq. (\ref{eq:ss})
as
\begin{eqnarray}
\sum_{i=s,s',s''}S(\rho_{i})-S(\rho_{123}) & \ge & \sum_{i=s,s'}S(\rho_{i})-S(\rho_{ss'})+\nonumber \\
 &  & \sum_{i=s,s''}S(\rho_{i})-S(\rho_{ss''}).
\end{eqnarray}
Or, equivalently,
\begin{eqnarray}
S(\rho_{123}||\rho_{s}\otimes\rho_{s'}\otimes\rho_{s''}) & \ge & S(\rho_{ss'}||\rho_{s}\otimes\rho_{s'})\nonumber \\
 &  & +S(\rho_{ss''}||\rho_{s}\otimes\rho_{s''}).
\end{eqnarray}
Using $S(\rho_{s'})-S(\rho_{s's''})=I(\rho_{s's''})-S(\rho_{s''})$,
$S(\rho_{s''})-S(\rho_{s's''})=I(\rho_{s's''})-S(\rho_{s'})$, and
the definition of multipartite mutual information presented in Sec.
\ref{sec1}, we see that the last equation together with the ESSA
implies the inequality in Eq. (\ref{eq:gmr1-i}), concluding thus
the proof of the proposition. \qed

\end{proof}

For three-partite systems whose reduced states are such that $S(\rho_{s'})\le S(\rho_{s's''})$
and $S(\rho_{s''})\le S(\rho_{s's''})$, we have a weaker version
of the inequality in Eq. (\ref{eq:gmr1-i}):
\begin{equation}
I(\rho_{123})\ge I(\rho_{ss'})+I(\rho_{ss''}).\label{eq:gmr1}
\end{equation}
This \emph{generalized monogamy relation} entail that the total correlation
in a three-particle mixed state restrict the information which a subsystem
can share individually with the other two parties of the system. Similar
constraints were obtained recently for bipartite quantum correlations
using the global quantum discord (see Ref. \cite{Braga-MGQD} and
references therein). For reviews about the classical and quantum aspects
of correlations see Refs. \cite{Celeri-IJQI,Modi-RMP,Maziero-bjp}.

We observe that Eq. (\ref{eq:gmr1-i}) is in general a stronger version
of the monogamy inequality in Eq. (\ref{eq:gmr1}). Although the first
is cumbersome, it shows that the three-partite mutual information
of any state $\rho_{123}$ limits the total amount of correlation
that its bipartite reductions can possess.

\subsection{An Inequality for Three-Partite Mutual Information}

\begin{proposition}

The total mutual information of three-partite mixed states is lower
bounded by the sum of the mutual informations of its bipartite marginals
as follows:
\begin{equation}
I(\rho_{123})\ge\frac{2}{3}\mathcal{I}_{2}(\rho_{123}).\label{eq:mr1}
\end{equation}

\end{proposition}

\begin{proof}

It is straightforward to prove this proposition by combining Eq. (\ref{eq:gmr1})
for $ss's''=123,231,312$. \qed

\end{proof}

In words, the lower bound in Eq. (\ref{eq:mr1}) means that the distinguishability
between $\rho_{123}$ and $\rho_{1}\otimes\rho_{2}\otimes\rho_{3}$ is
greater or equal than two-thirds of the sum of the distinguishabilities between
its bipartite marginals and their one-particle reductions in the product
form.

Now, we show that a subset of the classically correlated states (see
for instance the reference \cite{Piani-NLB}) saturates the inequality
above. These states can be written as 
$\chi_{123}=\sum_{i_{1}=i_{2}=i_{3}}p_{i_{1}i_{2}i_{3}}|\psi_{i_{1}}\psi_{i_{2}}\psi_{i_{3}}\rangle\langle\psi_{i_{1}}\psi_{i_{2}}\psi_{i_{3}}|,$
where $\{p_{i_{1}i_{2}i_{3}}\}$ is a probability distribution (that is to say, $p_{i_{1}i_{2}i_{3}}\ge0$
and $\sum_{i_{1}=i_{2}=i_{3}}p_{i_{1}i_{2}i_{3}}=1$) and $\{|\psi_{i_{s}}\rangle\}$
are local orthonormal basis. Noting that the entropies are given by
\begin{eqnarray}
S(\chi_{123})&=&S(\chi_{ss'})=S(\chi_{s''})\nonumber\\
&=&\sum_{i_{1}=i_{2}=i_{3}}p_{i_{1}i_{2}i_{3}}\log_{2}p_{i_{1}i_{2}i_{3}}^{-1},
\end{eqnarray}
with $ss'=12\mbox{, }13\mbox{, }23$ and $s''=1\mbox{, }2\mbox{, }3$,
one obtain $I(\chi_{ss'})=S(\chi_{s''})$ and thus $2\mathcal{I}_{2}(\chi_{123})/3=2S(\chi_{s''})=I(\chi_{123}),$
which shows that the class of states $\chi_{123}$ indeed saturates
the inequality in Eq. (\ref{eq:mr1}), and consequently lead to equality
in the SSA \cite{Hayden-SSvNE} and in the ESSA. So, the lower bound
in Eq. (\ref{eq:mr1}) is as tight as it could be.

On the other hand, one can show that for states with no genuine tripartite
correlation \cite{Bennett-GMC,Giorgi-GMC,Maziero-GMC}, viz., states
of the form $\rho_{s}\otimes\rho_{s's''}$, the equality holds: $I(\rho_{s}\otimes\rho_{s's''})=I(\rho_{s's''})
=\mathcal{I}_{2}(\rho_{s}\otimes\rho_{s's''})$
(see Sec. \ref{sec:residual}). In Sec. \ref{sec2a} we proved that
a similar equality is obtained for three-partite pure states, namely,
$I(|\psi_{123}\rangle)=\mathcal{I}_{2}(|\psi_{123}\rangle)$. In addition,
Streltsov \emph{et al.} proved that any positive measure of correlation
that is non-increasing under quantum operations on at least one of
its subsystems is maximal for some pure state \cite{Streltsov-Monogamy}.
As the quantum relative entropy fulfills this condition \cite{Vedral-RE-RMP},
we have that for any density operator $\rho_{123}$ there exists a
state vector $|\psi_{123}\rangle$ such that $I(|\psi_{123}\rangle)\ge I(\rho_{123})$.
As $I(|\psi_{123}\rangle)=\mathcal{I}_{2}(|\psi_{123}\rangle)$ for
all three-partite pure states, one may ask if the total mutual information
of general, mixed, three-partite states is limited by the correlations
of its bipartite marginals. Below we give an example showing that
this is not generally the case. Let us regard a mixture of $\mathrm{W}$
\cite{W-state} and $\mathrm{GHZ}$ \cite{GHZ} three-qubit states
\cite{footnote-1}:
\begin{equation}
\rho_{123}=p|\mathrm{W}_{3}\rangle\langle\mathrm{W}_{3}|+(1-p)|\mathrm{GHZ}_{3}\rangle\langle\mathrm{GHZ}_{3}|,\label{eq:W-GHZ}
\end{equation}
with $|\mathrm{W}_{3}\rangle=3^{-1/2}(|0_{1}0_{2}1_{3}\rangle+|0_{1}1_{2}0_{3}\rangle+|1_{1}0_{2}0_{3}\rangle)$
and $0\le p\le1$. The total correlation of $\rho_{123}$ and the
sum of the mutual informations of its bipartite marginals are shown
in Fig. \ref{fig1}. 

\begin{figure}
\includegraphics[scale=0.5]{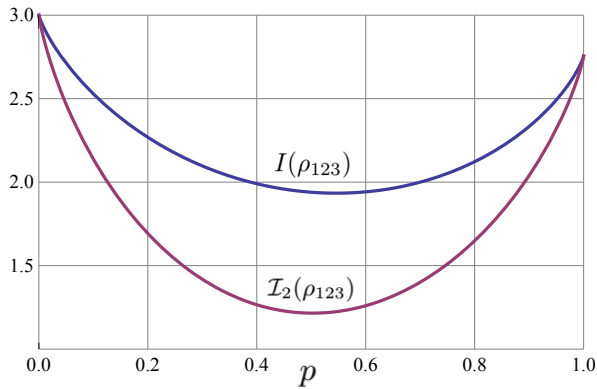}
\caption{(Color online). Total mutual information $I(\rho_{123})$ and the
sum of bipartite total correlations $\mathcal{I}_{2}(\rho_{123})$
for the mixture of $\mathrm{W}$ and $\mathrm{GHZ}$ states in Eq.
(\ref{eq:W-GHZ}) (see the text for details).}
\label{fig1}
\end{figure}

 We see that $I(\rho_{123})>\mathcal{I}_{2}(\rho_{123})$ for all
values of $p$ with exception of $p=0$ and $p=1$, where the states
are pure and we have equality of the correlations ($I(\rho_{123})=\mathcal{I}_{2}(\rho_{123})$).

If there are two states $|\psi_{123}\rangle$ and $\rho_{123}$ (i)
having the same bipartite reductions and (ii) connected by local quantum
operations, i.e., $\rho_{123}=\Lambda_{l}[|\psi_{123}\rangle]$, with
$\Lambda_{l}:\mathcal{D}(\mathcal{H}_{s})\rightarrow\mathcal{D}(\mathcal{H}_{s})$,
then
\begin{eqnarray}
I(\rho_{123})=I(\Lambda_{l}[|\psi_{123}\rangle]) & \le & I(|\psi_{123}\rangle)\nonumber \\
 & = & \mathcal{I}_{2}(|\psi_{123}\rangle)=\mathcal{I}_{2}(\rho_{123}).
\end{eqnarray}
So, the negative result above (Fig. \ref{fig1}) rules out the fulfillment
of both conditions (i) and (ii) for general pairs of states, though
for particular cases these conditions can be satisfied.

\subsection{Residual Versus Genuine Three-Partite Correlations}

\label{sec:residual}

Monogamy inequalities are frequently used as a starting point to defining measures for residual multipartite entanglement
\cite{Wootters-Dist.Ent} and quantum discord \cite{Braga-MGQD}. The inequality above (Eq. (\ref{eq:gmr1})) 
can be used to define a positive quantifier for the \emph{residual three-partite total correlation} in $\rho_{123}$
as follows:
\begin{equation}
I_{r}(\rho_{123}):=I(\rho_{123})-\frac{2}{3}\mathcal{I}_{2}(\rho_{123}).\label{eq:residual}
\end{equation}

\emph{Genuine $n$-partite correlations} (GnC) are those correlations that cannot be accounted for by looking at $n-1$ or 
less subsystems. It is natural questioning if $I_{r}$ does quantify GnC in three-partite 
states \cite{Bennett-GMC,Giorgi-GMC,Maziero-GMC}.
For the sake of answering this question, let us assume that a state
$\rho_{123}$ does not presents genuine three-partite total correlation,
e.g., $\rho_{123}=\rho_{1}\otimes\rho_{23}$. Hence $\rho_{12}=\rho_{1}\otimes\rho_{2}$
and $\rho_{13}=\rho_{1}\otimes\rho_{3}$. In this case $I(\rho_{123})=S(\rho_{1}\otimes\rho_{23}||\rho_{1}\otimes\rho_{2}\otimes\rho_{3})
=S(\rho_{23}||\rho_{2}\otimes\rho_{3})=I(\rho_{23}).$
As $I(\rho_{12})=I(\rho_{13})=0$, the residual tripartite correlation
is given by $I_{r}(\rho_{1}\otimes\rho_{23})=3^{-1}I(\rho_{23}).$
So, one see that although the residual total correlation somehow quantifies
the correlations in a multipartite state, $I_{r}$ is nonzero for
states which do not possess genuine three-partite total correlation.
So $I_{r}$ cannot be used for the purpose of quantifying or identifying
genuine multipartite correlations.

\section{Concluding remarks}

In this article we addressed the problem of distribution of mutual
information in multipartite systems, focusing on $n$-partite state
vectors and three-partite density operators. We obtained a general
relation for the relative entropy of total correlation of $n$-partite
pure states in terms of the mutual informations and entropies of its
marginals. The total correlation of three-partite pure states was
shown to be completely accounted for by the correlations of its bipartite
reductions. However, for systems in a pure state and with $n>3$ subsystems
the sum of the mutual informations of the $(n-1)$-partite reductions
of $|\psi_{1\dots n}\rangle$ overestimate its total correlation.
This fact indicates that, in this last case, there must exist redundant
information shared among the subsystems. That is to say, if correlation
is seen as shared information then a subsystem must share the same
information with two or more others parties of the whole physical
system.

Monogamy relations for bipartite correlations via bipartite correlations
were first noticed for entanglement measures \cite{Terhal-Monogamy,Wootters-Dist.Ent}
and for non-local quantum correlations \cite{Toner-Monogamy.NLQC}.
This kind of inequality was shown recently to be not generally applicable
for separable-state quantum correlations \cite{Streltsov-Monogamy}.
Here, continuing the program initiated in Ref. \cite{Braga-MGQD},
we showed that monogamy relations are restored for bipartite mutual
informations if we employ a relative entropy-based measure of total
multipartite correlation. These general monogamy inequalities led
to a tight lower bound for the mutual information of three-partite
mixed states in terms of the respective correlations of its bipartite
reductions. It is important to mention here that the main point of the original monogamy relations was to differentiate
quantum (non-separable or non-local) from classical correlations \cite{Terhal-Monogamy}.
As generalized monogamy relations via multipartite correlations hold
for total correlations, we notice that this kind of inequality cannot
be generally used to characterize the quantumness of the correlations in a physical system.

Looking for a possible interpretation for the equalities and inequalities
obtained in this article in terms of erasure of correlation by added
noise \cite{Groisman-QMI} is an interesting topic for future investigations.

\begin{acknowledgements}
This work was supported by the Conselho Nacional de Desenvolvimento Cient\'ifico e Tecnol\'ogico (CNPq) through the Instituto 
Nacional de Ci\^encia e Tecnologia de Informação Qu\^antica (INCQ-IQ).
\end{acknowledgements}

\appendix

\section{Proof of the equality in Eq. (\ref{QRE-equality})}
\label{appendixA}
Let us first prove some lemmas to be used subsequently.

\begin{lemma}
 Let $\chi_{s}\in\mathcal{H}_{s}$ and $\xi_{s'}\in\mathcal{H}_{s'}$ be any pair of density operators for the systems $s$ and $s'$, 
 respectively. It follows that
 \begin{equation}
  \log_{2}(\chi_{s}\otimes\xi_{s'})=\log_{2}(\chi_{s})\otimes\mathbb{I}_{s'} + \mathbb{I}_{s}\otimes\log_{2}(\xi_{s'}).
 \end{equation}
 \label{lemma1}
\end{lemma}
\begin{proof}
 A function $f:\mathbb{C}\rightarrow\mathbb{C}$ from the complex numbers to the complex numbers 
 is applied to normal operators $O$ (with eigenvalues $o_{i}$ and 
 eigenvectors $|o_{i}\rangle$) as follows: $f(O)=\sum_{i}f(o_{i})|o_{i}\rangle\langle o_{i}|$.
 Now let the eigen-decompositions of the density operators regarded above be $\chi_{s}=\sum_{i}a_{i}|a_{i}\rangle\langle a_{i}|$ and
 $\xi_{s}=\sum_{i}b_{i}|b_{i}\rangle\langle b_{i}|$. So
 \begin{eqnarray}
  \log_{2}(\chi_{s}\otimes\xi_{s'}) 
  %&=& \log_{2}(\sum_{i,j}a_{i}b_{j}|a_{i}\rangle\langle a_{i}|\otimes|b_{j}\rangle\langle b_{j}|) \nonumber \\
  &=& \sum_{i,j}\log_{2}(a_{i}b_{j})|a_{i}\rangle\langle a_{i}|\otimes|b_{j}\rangle\langle b_{j}|  \\
  &=& \sum_{i}\log_{2}(a_{i})|a_{i}\rangle\langle a_{i}|\otimes\sum_{j}|b_{j}\rangle\langle b_{j}| \nonumber \\
  && + \sum_{i}|a_{i}\rangle\langle a_{i}|\otimes\sum_{j}\log_{2}(b_{j})|b_{j}\rangle\langle b_{j}|  \\
  &=& \log_{2}(\chi_{s})\otimes\mathbb{I}_{s'} + \mathbb{I}_{s}\otimes\log_{2}(\xi_{s'}) 
 \end{eqnarray}
In order to obtain the last equality we used the completeness relations for the state spaces $\mathcal{H}_{s}$ and 
$\mathcal{H}_{s'}$. \qed
\end{proof}

\begin{lemma}
 Let $f:\mathbb{C}\rightarrow\mathbb{C}$ be any function from the complex numbers to the complex numbers and 
 let $\xi_{s'}\in\mathcal{H}_{s'}$ be any state of system $s'$. Then it follows that
 \begin{equation}
  f(\mathbb{I}_{s}\otimes\xi_{s'}) = \mathbb{I}_{s}\otimes f(\xi_{s'}).
 \end{equation}
 \label{lemma2}
\end{lemma}
\begin{proof}
 To prove this lemma we need only to use the eigen-decomposition of $\xi_{s'}=\sum_{i}b_{i}|b_{i}\rangle\langle b_{i}|$
 and the closure relation $\mathbb{I}_{s}=\sum_{i}|a_{i}\rangle\langle a_{i}|$ in order to write:
 \begin{eqnarray}
  f(\mathbb{I}_{s}\otimes\xi_{s'}) &=& \sum_{i,j}f(b_{j})|a_{i}\rangle\langle a_{i}|\otimes|b_{j}\rangle\langle b_{j}| \\
  &=& \sum_{i}|a_{i}\rangle\langle a_{i}|\otimes\sum_{j}f(b_{j})|b_{j}\rangle\langle b_{j}| \\
  &=& \mathbb{I}_{s}\otimes f(\xi_{s'}).
 \end{eqnarray}
 \qed
\end{proof}

\begin{lemma}
 Let $f:\mathbb{C}\rightarrow\mathbb{C}$ be any function from the complex numbers to the complex numbers,
 let $\chi_{ss'}\in\mathcal{H}_{ss'}$ be any global state for the systems $s$ and $s'$, and let $\xi_{s'}\in\mathcal{H}_{s'}$ 
 be any state for the system $s'$. Then it follows that
 \begin{equation}
 \mathrm{tr}(\chi_{ss'}f(\mathbb{I}_{s}\otimes\xi_{s'})) =  \mathrm{tr}_{s'}(\chi_{s'}f(\xi_{s'})),
 \end{equation}
 with $\chi_{s'}=\mathrm{tr}_{s}(\chi_{ss'})$.
 \label{lemma3}
\end{lemma}
\begin{proof}
 First we use a basis $\{|a_{i}\rangle\otimes|b_{j}\rangle\}$ for $\mathcal{H}_{ss'}$ to write
 \begin{eqnarray}
  \chi_{ss'} &=& \mathbb{I}_{ss'} \chi_{ss'} \mathbb{I}_{ss'} = 
  \sum_{i,j,k,l} \chi_{ss'}^{ijkl}|a_{i}\rangle\langle a_{k}|\otimes|b_{j}\rangle\langle b_{l}|,
 \end{eqnarray}
where we defined $\chi_{ss'}^{ijkl} = \langle a_{i}|\otimes \langle b_{j}|\chi_{ss'}|a_{i}\rangle\otimes|b_{j}\rangle$. 
For this global state, the reduced density operator of system $s'$ is
\begin{eqnarray}
 \chi_{s'} &=& \mathrm{tr}_{s} (\chi_{ss'}) =\mathrm{tr}_{s}(\sum_{i,j,k,l} \chi_{ss'}^{ijkl}|a_{i}\rangle\langle a_{k}|\otimes|b_{j}\rangle\langle b_{l}|) \\
 &=& \sum_{i,j,k,l} \chi_{ss'}^{ijkl}\mathrm{tr}_{s}(|a_{i}\rangle\langle a_{k}|)|b_{j}\rangle\langle b_{l}| %\\
% &=& \sum_{i,j,l} \chi_{ss'}^{ijil}|b_{j}\rangle\langle b_{l}|.
\end{eqnarray}
We now make use of these expressions and of Lemma \ref{lemma2} to get
\begin{eqnarray}
 && \mathrm{tr}(\chi_{ss'}f(\mathbb{I}_{s}\otimes\xi_{s'})) = \mathrm{tr}(\chi_{ss'}\mathbb{I}_{s}\otimes f(\xi_{s'})) \\
 && =  \mathrm{tr}_{s'}\mathrm{tr}_{s}(\sum_{i,j,k,l} \chi_{ss'}^{ijkl}|a_{i}\rangle\langle a_{k}|\otimes|b_{j}\rangle\langle b_{l}|f(\xi_{s'})) \\
 && =  \mathrm{tr}_{s'}(\sum_{i,j,k,l} \chi_{ss'}^{ijkl}\mathrm{tr}_{s}(|a_{i}\rangle\langle a_{k}|)|b_{j}\rangle\langle b_{l}|f(\xi_{s'})) \\
 && = \mathrm{tr}_{s'}(\chi_{s'}f(\xi_{s'})).
\end{eqnarray}
\qed
\end{proof}

Now we have the tools we need to prove the equality in Eq. (\ref{QRE-equality}).
\begin{proposition}
Let $\rho_{1\cdots n}\in\mathcal{D}(\mathcal{H}_{1\cdots n})$ be any $n$-partite state with marginals density operators $\rho_{s}$ 
for its subsystems $s=1,\cdots,n$. Let $\bigotimes_{s=1}^{n}\sigma_{s}$ be any $n$-partite product state. It follows that
\begin{equation}
S(\rho_{1\cdots n}||{\bigotimes_{s=1}^{n}}\sigma_{s})=S(\rho_{1\cdots n}||{\bigotimes_{s=1}^{n}}\rho_{s})
+{\sum_{s=1}^{n}}S(\rho_{s}||\sigma_{s}).
\end{equation}
\end{proposition}
\begin{proof}
Let us start by using the definition of quantum relative entropy in Eq. (\ref{eq:QRE}) to write
\begin{equation}
 S(\rho_{1\cdots n}||\bigotimes_{i=1}^{n}\sigma_{i})=-S(\rho_{1\cdots n})-\mathrm{tr}(\rho_{1\cdots n}\log_{2}\bigotimes_{i=1}^{n}\sigma_{i}).
 \label{eq:REE1}
 \end{equation}

Utilizing Lemmas \ref{lemma1} and \ref{lemma3} we can express the last term on the right hand side of the last equation as
\begin{eqnarray}
&& \mathrm{tr}(\rho_{1\cdots n}\log_{2}\bigotimes_{i=1}^{n}\sigma_{i}) =
\mathrm{tr}(\rho_{1\cdots n}\log_{2}(\sigma_{1})\otimes\mathbb{I}_{2}\otimes\cdots\otimes\mathbb{I}_{n}) \nonumber \\ 
&& + \cdots + \mathrm{tr}(\rho_{1\cdots n}\mathbb{I}_{1}\otimes\cdots\otimes\mathbb{I}_{n-1}\otimes\log_2(\sigma_{n})) \\
&&=  \mathrm{tr}(\rho_{1}\log_{2}(\sigma_{1})) + \cdots + \mathrm{tr}(\rho_{n}\log_2(\sigma_{n})) \\
&& = \sum_{i=1}^{n}\mathrm{tr}(\rho_{i}\log_{2}\sigma_{i}).
\label{eq:REE2}
\end{eqnarray}

In a similar manner, we have
\begin{equation}
\mathrm{tr}(\rho_{1\cdots n}\log_{2}\bigotimes_{i=1}^{n}\rho_{i}) = \sum_{i=1}^{n}\mathrm{tr}(\rho_{i}\log_{2}\rho_{i}).
\end{equation}

Using these two relations we get
\begin{eqnarray}
&& S(\rho_{1\cdots n}||\bigotimes_{i=1}^{n}\sigma_{i})=-S(\rho_{1\cdots n})
-\mathrm{tr}(\rho_{1\cdots n}\log_{2}\bigotimes_{i=1}^{n}\rho_{i}) \nonumber \\
&& + \sum_{i=1}^{n}\mathrm{tr}(\rho_{i}\log_{2}\rho_{i})
-\sum_{i=1}^{n}\mathrm{tr}(\rho_{i}\log_{2}\sigma_{i}) \\
&& = S(\rho_{1\cdots n}||\bigotimes_{i=1}^{n}\rho_{i}) +\sum_{i=1}^{n}S(\rho_{i}||\sigma_{i}),
\end{eqnarray}
concluding thus the proof of the proposition. \qed
\end{proof}

We observe that the last equality can be expressed as
\begin{equation}
S(\rho_{1\cdots n}||\bigotimes_{i=1}^{n}\sigma_{i}) = S(\rho_{1\cdots n}||\bigotimes_{i=1}^{n}\rho_{i})
 + S(\bigotimes_{i=1}^{n}\rho_{i}||\bigotimes_{i=1}^{n}\sigma_{i}).
\end{equation} 
So, in this case, the triangle inequality
 is satisfied (and saturated) by quantum relative entropy.

%%%%%%%%%%%%%%%%%%%%%%%%%%%%%%%%%%%%%%%%%%%%%%%%%%%%%%%%%%%%%%%%%%%%%%%%%%%%%%%%%%%%%%%%%%%%%%%%


\begin{thebibliography}{}
\bibitem{QRE} It is worthwhile mentioning that because, in general, $S(\rho||\sigma)\ne S(\sigma||\rho)$ 
and because quantum relative entropy (QRE) lacks a triangle inequality, it is not a well defined metric. 
Nevertheless, as the QRE is positive semidefinite, i.e., $S(\rho||\sigma)\ge 0$ for any pair of quantum states, 
with equality if and only if $\rho=\sigma$, and because 
the QRE is related to the probability of confusing (distinguishability of) two quantum states, 
$S(\rho||\sigma)$ can be seen as a quantifier for how close (or how distant) $\rho$ and $\sigma$ are from each other in the state space.
\bibitem{footnote-1} The states in Eq. (\ref{eq:W-GHZ}) are typical
system-environment states generated for a two-qubit system when one
of the qubits suffers decoherence due to the action of composite amplitude
and phase channels \cite{Maziero-GMC,Farias-ME-Decoherence}.
\end{thebibliography}
\end{document}